\newif\ifspringer
\springertrue
\ifspringer
\documentclass[twocolumn,natbib,envcountsame,numbook]{svjour3}
\pdfoutput=1
\smartqed
\else
\documentclass[abstract,DIV=calc,fleqn,twocolumn,10pt]{scrartcl}
\addtokomafont{title}{\normalfont\rmfamily}
\addtokomafont{disposition}{\rmfamily}
\addtokomafont{descriptionlabel}{\rmfamily}
\fi
\usepackage[utf8]{inputenc}
\usepackage{amsmath,amssymb}
\usepackage{mathptmx}
\ifspringer\else
\usepackage{natbib}
\setcitestyle{round,semicolon,aysep={}}
\usepackage{amsthm}
\usepackage{authblk}

\numberwithin{equation}{section}
\numberwithin{table}{section}
\fi
\usepackage{microtype}
\usepackage{tikz}
\usetikzlibrary{arrows,calc,decorations.pathreplacing}
\usepackage[textsize=tiny]{todonotes}
\usepackage{paralist}
\usepackage{xcolor}
\usepackage{booktabs}
\usepackage{hyperref}
\usepackage{mathtools}
\hypersetup{allbordercolors=[HTML]{0404ee},pdfborderstyle={/S/U/W 1}}

\newcommand{\loosness}{\ensuremath{\lambda}}
\newcommand{\overl}{\ensuremath{h}}
\newcommand{\slack}{\ensuremath{\sigma}}
\newcommand{\pt}{po\-ly\-nom\-i\-al-time}
\newcommand{\ppt}{pseu\-do-po\-ly\-nom\-i\-al-time}
\newcommand{\pp}{pseu\-do-po\-ly\-nom\-i\-al}

\newcommand{\fp}{fi\-xed-pa\-ra\-me\-ter}

\newcommand{\STW}[1]{\textsc{\ensuremath{#1}-Loose In\-ter\-val-Con\-strained Sched\-ul\-ing}} \newcommand{\ICSS}[1]{\textsc{\ensuremath{#1}-Slack In\-ter\-val-Con\-strained Sched\-ul\-ing}}

\newcommand{\ICS}{\textsc{In\-ter\-val-Con\-strained
 Sched\-ul\-ing}}

\newcommand{\amax}{a_{\text{max}}}
\newcommand{\binp}{\textsc{Bin Packing}}
\newcommand{\tpart}{\textsc{3-Partition}}

\newcommand{\npmptly}{non-pre\-emp\-tive\-ly}

\makeatletter \let\cl@chapter\relax \makeatother
\usepackage[capitalize]{cleveref}

\ifspringer
\spnewtheorem{construction}[theorem]{Construction}{\bfseries}{\normalfont}
\spnewtheorem{algorithm}[theorem]{Algorithm}{\bfseries}{\normalfont}
\else
\newtheorem{theorem}{Theorem}[section]

\theoremstyle{definition}

\newtheorem{lemma}[theorem]{Lemma}

\newtheorem{definition}[theorem]{Definition}
\newtheorem{algorithm}[theorem]{Algorithm}
\newtheorem{remark}[theorem]{Remark}
\newtheorem{construction}[theorem]{Construction}
\newtheorem{proposition}[theorem]{Proposition}
\fi

\crefname{construction}{Construction}{Constructions}
\crefname{algorithm}{Algorithm}{Algorithms}
\crefname{section}{Section}{Sections}
\crefname{figure}{Figure}{Figures}
\crefname{table}{Table}{Tables}
\newcommand{\deadl}{d}
\newcommand{\proct}{p}
\newcommand{\relt}{t}
\newcommand{\reltmax}{\relt_{\text{max}}}

\DeclareMathOperator\poly{poly}

\newcommand{\prob}[5]{%
  \begingroup
  \par\medskip
  \noindent \textsc{#1}\nopagebreak[4]
  \par\noindent\hangindent=0.5\parindent\textit{#2}  #3
  \par\noindent\hangindent=0.5\parindent\textit{#4}  #5
  \par  \medskip
  \endgroup
}

\newcommand{\decprob}[3]{\prob{#1}{Input:}{#2}{Question:}{#3}}

\title{A parameterized complexity view on \npmptly{} scheduling in\-ter\-val-con\-strained jobs: few machines, small looseness, and small slack}

\ifspringer
\titlerunning{A parameterized complexity view on \npmptly{} scheduling in\-ter\-val-con\-strained jobs}

\author{René van Bevern\thanks{René van Bevern is supported by grant 16-31-60007 mol\textunderscore{}a\textunderscore{}dk of the Russian Foundation for Basic Research (RFBR).} \and Rolf Niedermeier \and Ondřej Suchý\thanks{Ondřej Suchý is supported by grant 14-13017P of the Czech Science Foundation.}}
\date{Submitted: August 7, 2015\\Accepted: March 23, 2016}
\journalname{Journal of Scheduling}
\institute{René van Bevern \at Novosibirsk State University, Novosibirsk,  Russian Federation, \email{rvb@nsu.ru}
\and  Rolf Niedermeier \at Institut f\"ur Softwaretechnik und Theoretische Informatik,
  TU Berlin, Germany, \email{rolf.niedermeier@tu-berlin.de}
\and Ondřej Suchý \at Faculty of Information Technology, Czech Technical University in Prague, Prague, Czech Republic, \email{ondrej.suchy@fit.cvut.cz}}
\else
\author[1]{René van Bevern}
\author[2]{Rolf Niedermeier}
\author[3]{Ondřej Suchý}

\affil[1]{Novosibirsk State University, Novosibirsk, Russian Federation, \texttt{rvb@nsu.ru}}
\affil[2]{Institut f\"ur Softwaretechnik und Theoretische Informatik,
  TU Berlin, Germany, \texttt{rolf.niedermeier@tu-berlin.de}}
\affil[3]{Faculty of Information Technology, Czech Technical University in Prague, Prague, Czech Republic, \texttt{ondrej.suchy@fit.cvut.cz}}
\fi

\begin{document}
\maketitle

\begin{abstract}
  \looseness=-1\noindent We study the problem of \npmptly{} sched\-uling $n$~jobs, each job~$j$ with a release time~$\relt_j$, a deadline~$\deadl_j$, and a processing time~$\proct_j$, on $m$~parallel identical machines.  \citet{CEHBW04} considered the two constraints $|\deadl_j-\relt_j|\leq \loosness{}\proct_j$ and $|\deadl_j-\relt_j|\leq\proct_j +\slack$ and showed the problem to be NP-hard for any~$\loosness>1$ and for any~$\slack\geq 2$.  We complement their results by parameterized complexity studies: we show that, for any~$\loosness>1$, the problem remains weakly NP-hard even for~$m=2$ and strongly W[1]-hard parameterized by~$m$.  We present a \ppt{} algorithm for constant~$m$ and~$\loosness$ and a
\fp{} tractability result for the parameter~$m$ combined with~$\slack$. %
\ifspringer
\keywords{%
\else
\paragraph{Keywords}\newcommand{\and}{~$\cdot$ }
\fi
release times and deadlines\and machine minimization\and sequencing within intervals\and shiftable intervals\and \fp{} tractability\and NP-hard problem
\ifspringer
}
\fi
\end{abstract}

\section{Introduction}

Non-preemptively scheduling jobs with release times and deadlines on a minimum number of machines is a well-stud\-ied problem both in offline and online variants \citep{CMS16,CGKN04,CEHBW04,MN07,Sah13}.  In its decision version, the problem is formally defined as follows:
\decprob{\ICS{}}
{A set~$J:=\{1,\dots,n\}$ of jobs, a number~$m\in\mathbb N$ of machines, each job~$j$ with a \emph{release time}~$\relt_j\in\mathbb N$, a \emph{deadline}~$\deadl_j\in\mathbb N$, and a \emph{processing time~$\proct_j\in\mathbb N$}.}
{Is there a schedule that schedules all jobs onto $m$~parallel identical machines such that
  \begin{compactenum}
  \item each job~$j$ is executed \npmptly{} for $\proct_j$~time units,
  \item each machine executes at most one job at a time, and
  \item each job~$j$ starts no earlier than~$\relt_j$ and is finished by~$\deadl_j$.
  \end{compactenum}
}
\noindent For a job~$j\in J$, we call the half-open interval~$[\relt_j,\deadl_j)$ its \emph{time window}.  A job may only be executed during %
its time window.  The \emph{length} of the time window is~$\deadl_j-\relt_j$.

We study \ICS{} with two additional constraints introduced by \citet{CEHBW04}.  These constraints relate the time window lengths of jobs to their processing times:

\paragraph{Looseness} If all jobs~$j\in J$ satisfy $|\deadl_j-\relt_j|\leq\loosness p_j$ for some number~$\loosness\in\mathbb R$, then the instance has \emph{looseness~\loosness{}}.   By \STW{\loosness{}} we denote the problem  restricted to instances of looseness~$\loosness$.

\paragraph{Slack} If all jobs~$j\in J$ satisfy $|\deadl_j-\relt_j|\leq \proct_j + \slack$ for some number~$\slack\in\mathbb R$, then the instance has \emph{slack~\slack{}}.  By \ICSS{\slack} we denote the problem restricted to instances of slack~$\slack$.

\bigskip\noindent Both constraints on \ICS{} are very natural: clients may accept some small deviation of at most~$\slack{}$ from the desired start times of their jobs.  Moreover, it is conceivable that clients allow for a larger deviation for jobs that take long to process anyway, leading to the case of bounded looseness~$\loosness$.  

\citet{CEHBW04} showed that, even for constant~$\loosness>1$ and constant~$\slack\geq 2$, the problems \STW{\loosness} and \ICSS{\slack} are strongly NP-hard. 

Instead of giving up on finding optimal solutions and resorting to approximation algorithms \citep{CGKN04,CEHBW04}, we conduct a more fine-grained complexity analysis of these problems employing the framework of \emph{parameterized complexity theory} \citep{CFK+15,DF13,FG06,Nie06}, which so far received comparatively little attention in the field of scheduling with seemingly only a handful of publications \citep{BJH+15,BMNW15,BF95,CEHBW04,FM03,HK06,HKS+15,MW15}.  In particular, we investigate the effect of the parameter~$m$ of available machines on the parameterized complexity of interval-constrained scheduling without preemption.

\begin{table*}
  \centering\ifspringer\else\small\fi
  \caption{Overview of results on \ICS{} for various parameter combinations.  The parameterized complexity with respect to the combined parameter~$\loosness + \slack$ remains open.}

  \begin{tabular}{rp{3cm}p{3cm}p{8cm}} \toprule%
    Combined & \multicolumn{3}{c}{Parameter}\\ \cmidrule(r){2-4}%
    with & looseness~$\loosness$ & slack~$\slack$ & number~$m$ of machines\\
    \cmidrule(r){1-1} \cmidrule(r){2-4} %
    $\loosness$ & NP-hard for any~$\loosness>1$ \citep{CEHBW04}
                & \multicolumn{1}{p{3cm}}{\vspace{0.35\baselineskip}\hspace{1.25cm}?}
                & W[1]-hard for parameter~$m$ for any~$\loosness>1$ (\cref{w1hard}),\newline
weakly NP-hard for $m=2$ and any $\lambda>1$ (\cref{w1hard}),\newline
\pp{} time for fixed~$m$ and~$\loosness$ (\cref{anyalphafpt})
                   \\ \\
    $\slack$ & 
             & NP-hard for any~$\slack\geq 2$ \citep{CEHBW04}
               & fixed-parameter tractable for parameter \(\sigma+m\) (\cref{fpt}) \\ \\
    $m$ & & & NP-hard for~$m=1$ \citep{GJ79} \\
    \bottomrule
  \end{tabular}
  \label{tab:results}
\end{table*}

\paragraph{Related work}
\ICS{} is a classical scheduling problem and strongly NP-hard already on one machine \citep[problem~SS1]{GJ79}.
Besides the task of scheduling all jobs on a minimum number of machines, the literature contains a wide body of work concerning the maximization of the number of scheduled jobs on a bounded number of machines \citep{KLPS07}.  

For the objective of minimizing the number of machines, \citet{CGKN04} developed a factor-$O(\sqrt{\log n/\log\log n})$-ap\-prox\-i\-ma\-tion algorithm. \citet{MN07} formalized machine minimization and other objectives in terms of optimization problems in shiftable interval graphs.  Online algorithms for %
minimizing the number of machines
have been studied as well and we refer to recent work by \citet{CMS16} for an overview.

Our work refines the following results of \citet{CEHBW04}, who considered \ICS{} with bounds on the looseness and the slack.  They showed that \ICS{} is strongly NP-hard for any looseness~$\loosness>1$ and any slack~$\slack\geq 2$.  Besides giving approximation algorithms for various special cases, they give a \pt{} algorithm for~$\slack=1$ and a \fp{} tractability result for the combined parameter~$\slack$ and~$\overl{}$, where $\overl{}$~is the maximum number of time windows overlapping in any point in time.

\paragraph{Our contributions}  We analyze the parameterized complexity of \ICS{} with respect to three parameters: the number~$m$ of machines, the looseness~$\loosness{}$, and the slack~$\slack{}$.  More specifically, we refine known results of \citet{CEHBW04} using tools of parameterized complexity analysis.  An overview is given in \cref{tab:results}.

\looseness=-1 In \cref{sec:hardness},  we show that, for any~$\loosness>1$, \STW{\loosness} remains weakly NP-hard even on $m=2$~machines and that it is strongly W[1]-hard when parameterized by the number~$m$ of machines. %
In \cref{sec:l}, we give a \ppt{} algorithm for \STW{\loosness} for each fixed~$\loosness$ and~$m$.  %
Finally, in \cref{sec:s}, we give a \fp{} algorithm for \ICSS{\slack} when parameterized by~$m$ and~$\slack$.  This is in contrast to our result from \cref{sec:hardness} that the parameter combination~$m$ and~$\loosness$ presumably does not give \fp{} tractability results for \STW{\loosness}.

\section{Preliminaries}\label{sec:realpreliminaries}

\paragraph{Basic notation}  We assume that $0\in\mathbb N$.  For two vectors $\vec u=\allowbreak(u_1,\dots,u_k)$ and $\vec v=(v_1,\dots,v_k)$, we write $\vec u\leq \vec v$ if $u_i\leq v_i$ for all~$i\in\{1,\dots,k\}$.  Moreover, we write $\vec u\lneqq \vec v$ if $\vec u\leq \vec v$ and $\vec u\ne \vec v$, that is, $\vec u$ and $\vec v$~differ in at least one component.  Finally, $\vec 1^k$~is the $k$-dimensional vector consisting of $k$~1-entries.

\paragraph{Computational complexity} We assume familiarity with the basic concepts of NP-hardness and \pt{} many-one reductions \citep{GJ79}.  We say that a problem is \emph{(strongly) $C$-hard} for some complexity class~$C$ if it is $C$-hard even if all integers in the input instance are bounded from above by a polynomial in the input size.  Otherwise, we call it \emph{weakly $C$-hard}.

In the following, we introduce the basic concepts of parameterized complexity theory, which are in more detail discussed in corresponding text books \citep{CFK+15,DF13,FG06,Nie06}.

\paragraph{Fixed-parameter algorithms} The idea in \fp{} algorithmics is to accept exponential running times, which are seemingly inevitable in solving NP-hard problems, but to restrict them to one aspect of the problem, the \emph{parameter}.  

Thus, formally, an instance of a \emph{parameterized problem}~\(\Pi\) is a pair~$(x,k)$ consisting of the input~$x$ and the parameter~$k$.  A parameterized problem~$\Pi$ is \emph{\fp{} tractable~(FPT)} with respect to a parameter~$k$ if there is an algorithm solving any instance of~$\Pi$ with size~$n$ in $f(k) \cdot \poly(n)$~time for some computable function~$f$.  Such an algorithm is called a \emph{\fp{} algorithm}. It is potentially efficient for small values of~$k$, in contrast to an algorithm that is merely running in polynomial time for each fixed~$k$ (thus allowing the degree of the polynomial to depend on~$k$).  FPT is the complexity class of \fp{} tractable parameterized problems.  

We refer to the sum of parameters~$k_1+k_2$ as the \emph{combined parameter~$k_1$ and~$k_2$}.  %

\paragraph{Parameterized intractability} To show that a problem is presumably not \fp{} tractable, there is a parameterized analog of NP-hardness theory.  The parameterized analog of NP is the complexity class W[1]\({}\supseteq{}\)FPT, where it is conjectured that FPT\({}\ne{}\)W[1].  A parameterized problem~$\Pi$ with parameter~\(k\) is called \emph{W[1]-hard} if $\Pi$~being \fp{} tractable implies W[1]\({}={}\)FPT.  W[1]-hardness can be shown using a parameterized reduction from a known W[1]-hard problem: a \emph{parameterized reduction} from a parameterized problem~$\Pi_1$ to a parameterized problem~$\Pi_2$ is an algorithm mapping an instance~$I$ with parameter~$k$ to an instance~$I'$ with parameter~$k'$ in time~$f(k)\cdot\poly(|I|)$ such that $k'\leq g(k)$ and $I'$~is a yes-instance for~$\Pi_1$ if and only if $I$~is a yes-instance for~$\Pi_2$, where \(f\)~and~\(g\) are arbitrary computable functions.

\section{A strengthened hardness result}\label{sec:hardness}

In this section, we strengthen a hardness result of \citet{CEHBW04}, who showed that \STW{\loosness{}} is NP-hard for any~$\loosness{}>1$.   This section proves the following theorem:

\begin{theorem}\label[theorem]{w1hard}
  Let $\loosness{}\colon\mathbb N\to\mathbb R$~be such that $\loosness{}(n)\geq 1+n^{-c}$ for some integer~$c\geq 1$ and all~$n\geq 2$.

  Then \STW{\loosness{}(n)} of $n$~jobs on $m$~machines is
  \begin{enumerate}[(i)]
  \item\label{w1hard1} weakly NP-hard for~$m=2$, and
  \item\label{w1hard3} strongly W[1]-hard for parameter~$m$.
  \end{enumerate}
\end{theorem}

\noindent Note that \cref{w1hard}, in particular, holds for any constant function~$\lambda(n)>1$.  

We remark that \cref{w1hard} cannot be proved using the NP-hardness reduction given by \citet{CEHBW04}, which reduces \textsc{3-Sat} instances with \(k\)~clauses to \ICS{} instances with \(m=3k\)~machines.  Since \textsc{3-Sat} is trivially fixed-parameter tractable for the parameter number~\(k\) of clauses, the reduction of \citet{CEHBW04} cannot yield \cref{w1hard}.

Instead, to prove \cref{w1hard}, we give a parameterized \pt{} ma\-ny-one reduction from \binp{} with $m$~bins and $n$~items to \STW{\loosness{}(mn)} with $m$~machines and $mn$~jobs.
\decprob{\binp{}}
{A bin volume~$V\in\mathbb N$, a list~$a_1,\dots,a_n\in\mathbb N$ of items, and a number~$m\leq n$ of bins.}
{Is there a partition~$S_1\uplus\dots\uplus S_m=\{1,\dots,n\}$ such that $\sum_{i\in S_k}a_i\leq V$ for all~$1\leq k\leq m$?}
\noindent Since \binp{} is weakly NP-hard for $m=2$~bins and W[1]-hard parameterized by~$m$ even if all input numbers are polynomial in~$n$ \citep{JKMS13}, \cref{w1hard} will follow. 

Our reduction, intuitively, works as follows: for each of the \(n\)~items~\(a_i\) in a \binp{} instance with \(m\)~bins of volume~\(V\), we create a set~\(J_i:=\{j^1_i,\dots,j^m_i\}\) of \(m\)~jobs that have to be scheduled on \(m\)~mutually distinct machines.  Each machine represents one of the \(m\)~bins in the \binp{} instance.  Scheduling job~\(j_i^1\) on a machine~\(k\) corresponds to putting item~\(a_i\) into bin~\(k\) and will take \(B+a_i\)~time of machine~\(k\), where \(B\)~is some large integer chosen by the reduction.  If \(j_i^1\)~is not scheduled on machine~\(k\), then a job in~\(J_i\setminus\{j_i^1\}\) has to be scheduled on machine~\(k\), which will take only \(B\)~time of machine~\(k\).  Finally, we choose the latest deadline of any job as~\(nB+V\). Thus, since all jobs have to be finished by time~\(nB+V\) and since there are \(n\)~items, for each machine~\(k\), the items~\(a_i\) for which \(j_i^1\) is scheduled on machine~\(k\) must sum up to at most~\(V\) in a feasible schedule. This corresponds to satisfying the capacity constraint of~\(V\) of each bin.  

Formally, the reduction works as follows and is illustrated in \cref{fig:constrw1}.

\tikzstyle{interval}=[{[-)}, thick]
\tikzstyle{milestone}=[dotted]
\tikzstyle{process}=[thick, double, double distance=2pt, line cap=rect, shorten >=1.8pt, shorten <=1.8pt]
\begin{figure*}
  \centering
  \includegraphics{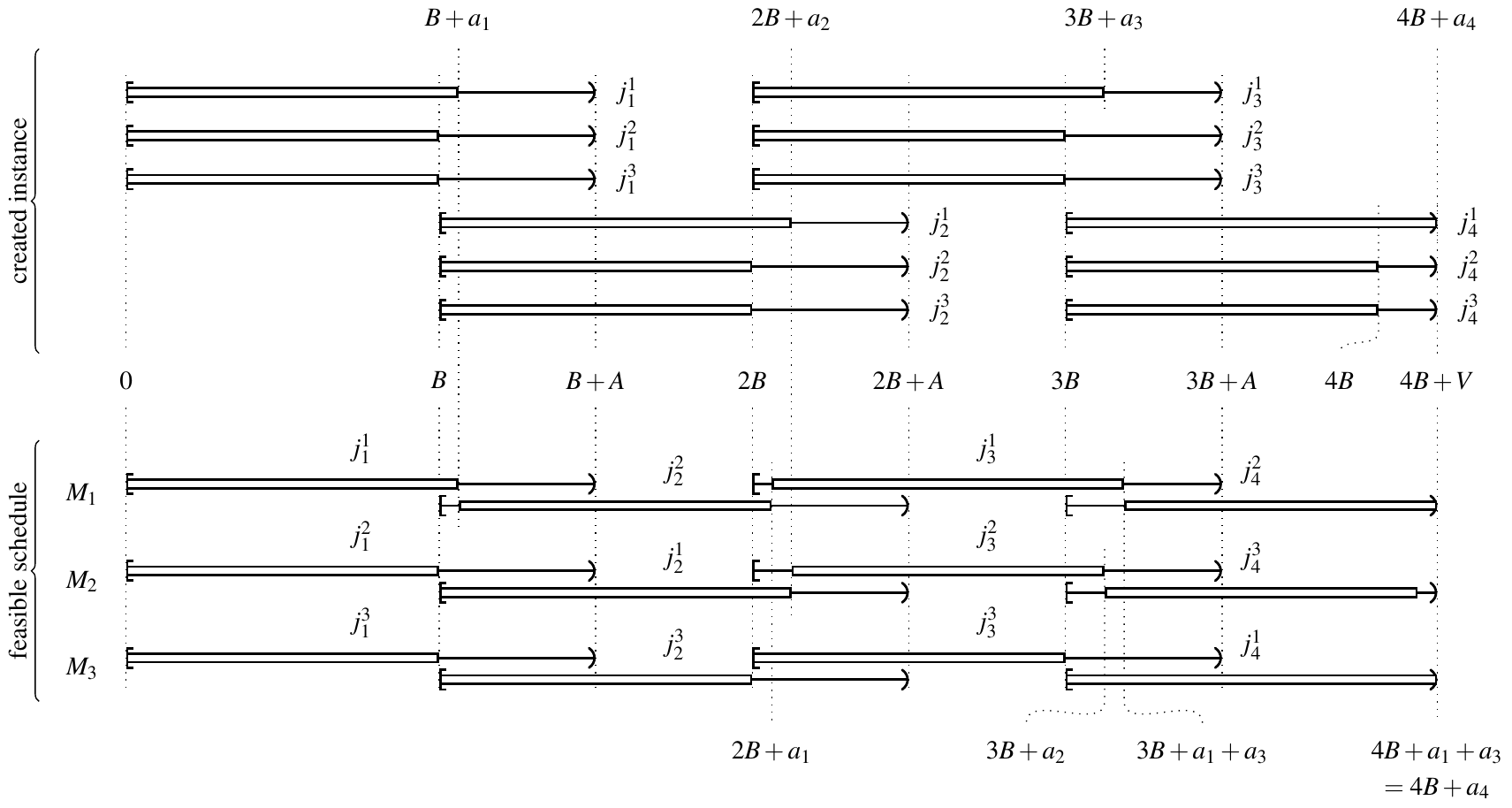}
  \caption{Reduction from \binp{} with four items~\(a_1=1,a_2=a_3=2,a_4=3\), bin volume~\(V=3\), and \(m=3\)~bins to \STW{3/2}. That is, \cref{constrw1} applies with \(c=1\), \(A=8\), and \(B=3\cdot 4\cdot 8=96\).  The top diagram shows (not to scale) the jobs created by \cref{constrw1}.  Herein, the processing time of each job is drawn as a rectangle of corresponding length in an interval being the job's time window.  The bottom diagram shows a feasible schedule for three machines~$M_1,M_2$, and~$M_3$ that corresponds to putting items~\(a_1\) and~\(a_3\) into the first bin, item~\(a_2\) into the second bin, and~\(a_4\) into the third bin.}
  \label{fig:constrw1}
\end{figure*}

\begin{construction}\label[construction]{constrw1}
  Given a \binp{} instance~$I$ with $n\geq 2$~items $a_1,\dots,a_n$ and $m\leq n$~bins, and $\loosness{}\colon\mathbb N\to\mathbb R$ such that $\loosness{}(n)\geq 1+n^{-c}$ for some integer~$c\geq 1$ and all~$n\geq 2$, we construct an \ICS{} instance with $m$~machines and $mn$~jobs as follows.  First, let
  \begin{align*}
    A&:=\sum_{i=1}^n a_i&&\text{and}&B&:=(mn)^c\cdot A\geq 2A.
  \end{align*}
  If $V>A$, then $I$~is a yes-instance of \binp{} and we return a trivial yes-instance of \ICS{}.

  Otherwise, we have $V\leq A$ and construct an instance of \ICS{} as follows: for each~$i \in \{1, \ldots, n\}$, we introduce a set~$J_i:=\{j^1_i,\dots,j^m_i\}$ of jobs.  For each job~$j\in J_i$, we choose the release time
  \begin{align}
    \relt_{j}&:=(i-1)B,\notag{}\\
    \intertext{the processing time}
    \label{bpconstr}\proct_j&:=\begin{cases}
      B+a_i&\text{if $j=j_i^1$,}\\
      B&\text{if $j\ne j_i^1$,}
    \end{cases}\\
    \intertext{and the deadline}
    \deadl_{j}&:=
                \begin{cases}
                  iB+A&\text{if $i<n$},\\
                  iB+V&\text{if $i=n$}.
                \end{cases}\notag{}
  \end{align}
  This concludes the construction.\qed
\end{construction}

\begin{remark}
  \cref{constrw1} outputs an \ICS{} instance with \emph{agreeable deadlines}, that is, the deadlines of the jobs have the same relative order as their release times.  Thus, in the offline scenario, all hardness results of \cref{w1hard} will also hold for instances with agreeable deadlines.

  In contrast, agreeable deadlines make the problem significantly easier in the online scenario: \citet{CMS16} showed an online-algorithm with constant competitive ratio for \ICS{} with agreeable deadlines, whereas there is a lower bound of~$n$ on the competitive ratio for general instances \citep{Sah13}.
\end{remark}

\noindent In the remainder of this section, we show that \cref{constrw1} is correct and satisfies all structural properties that allow us to derive \cref{w1hard}. %

First, we show that \cref{constrw1} indeed creates an \ICS{} instance with small looseness. 
\begin{lemma}\label[lemma]{indeed 1+e}
  Given a \binp{} instance with $n\geq 2$~items and $m$~bins, \cref{constrw1} outputs an \ICS{} instance with
  \begin{enumerate}[(i)]
  \item\label{1+e1} at most $m$~machines and $mn$~jobs and
  \item\label{1+e2} looseness~$\loosness(mn)$.
  \end{enumerate}
\end{lemma}

\begin{proof}
It is obvious that the output instance has at most $mn$~jobs and $m$~machines and, thus, \eqref{1+e1} holds.  
  
  Towards \eqref{1+e2}, observe that $mn\geq n\geq 2$, and hence, for each~$i\in\{1,\dots,n\}$ and each job~$j\in J_i$, \eqref{bpconstr}  yields
  \begin{align*}
      \frac{|\deadl_j-\relt_j|}{\proct_j}&{}\leq \frac{(iB+A)-(i-1)B}{B}
                                             {}=\frac{B+A}{B}
                                             {}= 1+\frac AB\\
      &{}= 1+\frac{A}{(mn)^c\cdot A}= 1+(mn)^{-c}\leq \loosness{}(mn).\ifspringer\hspace{1cm}\qed\else\qedhere\fi
  \end{align*}
\end{proof}

\noindent We now show that \cref{constrw1} runs in polynomial time and that, if the input \binp{} instance has polynomially bounded integers, then so has the output \ICS{} instance.

\begin{lemma}\label[lemma]{everything poly}
  Let $I$~be a \binp{} instance with $n\geq 2$ items~$a_1,\allowbreak\dots,a_n$ and let $\amax:=\max_{1\leq i\leq n}a_i$. \cref{constrw1} applied to~$I$
  \begin{enumerate}[(i)]
  \item\label{poly1} runs in time polynomial in~$|I|$ and
  \item\label{poly2}  outputs an \ICS{} instance whose release times and deadlines are bounded by a polynomial in $n+\amax$.
  \end{enumerate}
\end{lemma}

\begin{proof}
  We first show \eqref{poly2}, thereafter we show \eqref{poly1}.

  \eqref{poly2} It is sufficient to show that the numbers~$A$ and~$B$ in \cref{constrw1} are bounded polynomially in~$n+\amax$ since all release times and deadlines are computed as sums and products of three numbers not larger than~$A$, $B$, or~$n$.  Clearly, $A=\sum_{i=1}^na_i\leq n\cdot\max_{1\leq i\leq n}a_i$, which is polynomially bounded in~$n+\amax$.  Since~$mn\leq n^2$, also $B=(mn)^c\cdot A$~is polynomially bounded in~$n+\amax$.

  \eqref{poly1} The sum~$A=\sum_{1=1}^na_i$ is clearly computable in time polynomial in the input length.  It follows that also $B=(mn)^c\cdot A$ is computable in polynomial time.
\ifspringer\qed\fi
\end{proof}

\noindent It remains to prove that \cref{constrw1} maps yes-instances of \binp{} to yes-instances of \ICS{}, and no-instances to no-instances.

\begin{lemma}\label[lemma]{constrw1 correct}
  Given a \binp{} instance~$I$ with $m$~bins and the items $a_1,\dots,a_n$, \cref{constrw1} outputs an \ICS{} instance~$I'$ that is a yes-instance if and only if $I$~is.
\end{lemma}

\begin{proof}
  ($\Rightarrow$) Assume that $I$~is a yes-instance for \binp{}.  Then, there is a partition~$S_1\uplus \dots\uplus S_m=\{1,\dots,n\}$ such that $\sum_{i\in S_k} a_i \leq V$ for each~$k\in\{1,\dots,m\}$.  We construct a feasible schedule for~$I'$ as follows.  For each~$i \in \{1, \ldots, n\}$ and $k$~such that $i\in S_k$, we schedule~$j^1_i$ on machine~$k$ in the interval
\begin{align*}
  \Biggl[(i-1)B+\smashoperator{\sum_{j\in S_k, j < i}} a_j\quad,\quad iB+\smashoperator{\sum_{j \in S_k, j < i}} a_j+a_i\Biggr)
\end{align*}
and each of the $m-1$~jobs~$J_i\setminus\{j_i^1\}$ on a distinct machine~$\ell\in\{1,\dots,m\}\setminus\{k\}$ in the interval
\begin{align*}
  \Biggl[(i-1)B+\smashoperator{\sum_{j\in S_\ell, j < i}} a_j\quad,\quad iB+\smashoperator{\sum_{j \in S_\ell, j < i}} a_j\Biggr).
\end{align*}
It is easy to verify that this is indeed a feasible schedule.

($\Leftarrow$) Assume that $I'$~is a yes-instance for \ICS{}.  Then, there is a feasible schedule for~$I'$.  We define a partition~$S_1\uplus\dots\uplus S_m=\{1,\dots,n\}$ for~$I$ as follows.  For each $k\in\{1,\dots,m\}$, let
\begin{align}
  S_k&:=\{i\in\{1,\dots,n\}\mid \text{$j_i^1$ is scheduled on machine~$k$}\}.\label{skdef}
\end{align}
Since, for each $i\in\{1,\dots,n\}$, the job~$j_i^1$ is scheduled on exactly one machine, this is indeed a partition.  We show that $\sum_{i\in S_k} a_i \leq V$ for each $k\in\{1,\dots,m\}$.  Assume, towards a contradiction, that there is a~$k$ such that
\begin{align}
\sum_{i\in S_k}a_i> V.\label{overfilled}  
\end{align}
By \eqref{bpconstr}, for each $i \in \{1, \ldots, n\}$, the jobs in~$J_i$ have the same release time, each has processing time at least~$B$, and the length of the time window of each job is at most~$B+A\leq B+B/2 < 2B$.  Thus, in any feasible schedule, the execution times of the $m$~jobs in~$J_i$ mutually intersect.  Hence, the jobs in~$J_i$ are scheduled on $m$~mutually distinct machines.  By the pigeonhole principle, for each~$i\in\{1,\ldots,n\}$, exactly one job~$j_i^*\in J_i$ is scheduled on machine~$k$.  %
We finish the proof by showing that,
\begin{align}
   \label{slowclaim}\text{\parbox[t]{7cm}{$\forall i\in\{1,\dots,n\}$, job~$j_i^*$ is not finished before time~$iB+\smashoperator{\sum_{j \in S_k,j \leq i}} a_j$.}}
\end{align}
This claim together with \eqref{overfilled} then yields that job~$j_n^*$ is not finished before
\begin{align*}
  nB+\smashoperator{\sum_{j \in S_k,j \leq n}} a_j=  nB+\smashoperator{\sum_{j \in S_k}} a_j>nB+V,
\end{align*}
which contradicts the schedule being feasible, since jobs in~$J_n$ have deadline~$nB+V$ by \eqref{bpconstr}.  It remains to prove \eqref{slowclaim}.  We proceed by induction.

The earliest possible execution time of~$j_1^*$ is, by \eqref{bpconstr}, time~$0$.  The processing time of~$j_1^*$ is~$B$ if $j_1^*\ne j_1^1$, and $B+a_1$ otherwise.  By \eqref{skdef}, $1\in S_k$ if and only if $j_1^1$ is scheduled on machine~$k$, that is, if and only if $j_1^*=j_1^1$.  Thus, job~$j_1^*$ is not finished before $B+\sum_{j\in S_k,j\leq i}a_j$ and \eqref{slowclaim} holds for~$i=1$.  Now, assume that \eqref{slowclaim} holds for~$i-1$.  We prove it for~$i$.  Since $j_{i-1}^*$ is not finished before $(i-1)B+\sum_{j\in S_k,j\leq i-1}a_j$, this is the earliest possible execution time of~$j_i^*$. The processing time of~$j_i^*$ is~$B$ if $j_i^*\ne j_i^1$ and $B+a_i$ otherwise.  By \eqref{skdef}, $i\in S_k$ if and only if $j_i^*=j_i^1$.  Thus, job~$j_i^*$ is not finished before $iB+\sum_{j\in S_k,j\leq i}a_j$ and \eqref{slowclaim} holds.
\ifspringer\qed\fi
\end{proof}

\noindent We are now ready to finish the proof of \cref{w1hard}.

\begin{proof}[\ifspringer\else{}Proof \fi{}of \cref{w1hard}]
  By \cref{indeed 1+e,everything poly,constrw1 correct}, \cref{constrw1} is a polynomial-time many-one reduction from \binp{} with $n\geq 2$~items and $m$~bins to \STW{\loosness{}(mn)}, where $\loosness{}\colon\mathbb N\to\mathbb R$ such that $\loosness{}(n)\geq 1+n^{-c}$ for some integer~\(c\geq 1\) and all~$n\geq 2$.  We now show the points \eqref{w1hard1} and \eqref{w1hard3} of \cref{w1hard}.

  \eqref{w1hard1} follows since \binp{} is weakly NP-hard for~$m=2$ \citep{JKMS13} and since, by \cref{indeed 1+e}\eqref{1+e1}, \cref{constrw1} outputs instances of \STW{\loosness{}(mn)} with $m$~machines.

  \eqref{w1hard3} follows since \binp{} is W[1]-hard parameterized by~$m$ even if the sizes of the $n$~items are bounded by a polynomial in~$n$ \citep{JKMS13}.  In this case, \cref{constrw1} generates \STW{\loosness{}(mn)} instances for which all numbers are bounded polynomially in the number of jobs by \cref{everything poly}\eqref{poly2}.  Moreover, \cref{constrw1} maps the $m$~bins of the \binp{} instance to the $m$~machines of the output \ICS{} instance.
\ifspringer\qed\fi
\end{proof}

\noindent Concluding this section, it is interesting to note that \cref{w1hard} also shows W[1]-hardness of \STW{\loosness} with respect to the \emph{height} parameter considered by \citet{CEHBW04}:

\begin{definition}[Height]\label[definition]{def-st}
  For an \ICS{} instance and any time~$t\in\mathbb N$, let
\[
S_t:=\{j\in J\mid t\in[\relt_j,\deadl_j)\}
\]
denote the set of jobs whose time window contains time~$t$.  The \emph{height} of an instance is
\[
\overl{}:=\max_{t\in\mathbb N}|S_t|.
\]
\end{definition}

\begin{proposition}\label[proposition]{prop:hardh}
  Let $\loosness{}\colon\mathbb N\to\mathbb R$~be such that $\loosness{}(n)\geq 1+n^{-c}$ for some integer~$c\geq 1$ and all~$n\geq 2$.

  Then \STW{\loosness{}(n)} of $n$~jobs on $m$~machines is W[1]-hard parameterized by the height~$\overl{}$.
\end{proposition}

\begin{proof}
  \cref{prop:hardh} follows in the same way as \cref{w1hard};  one additionally has to prove that \cref{constrw1} outputs \ICS{} instances of height at most~$2m$.
  To this end, observe that, by~\eqref{bpconstr}, for each $i\in\{1,\dots,n\}$, there are $m$~jobs released at time~$(i-1)B$ whose deadline is no later than~$iB+A<(i+1)B$ since \(A\leq B/2\).
  These are all jobs created by \cref{constrw1}.
  Thus, $S_t$ contains only the $m$~jobs released at time~$\lfloor t/B\rfloor\cdot B$ and the $m$~jobs released at time~$\lfloor t/B-1\rfloor\cdot B$, which are $2m$~jobs in total.  \ifspringer\qed\fi
\end{proof}

\begin{remark}
  \cref{prop:hardh} complements findings of \citet{CEHBW04}, who provide a \fp{} tractability result for \ICS{} parameterized by~$\overl{}+\slack$: our result shows that their algorithm presumably cannot be improved towards a \fp{} tractability result for \ICS{} parameterized by~$\overl{}$ alone.
\end{remark}

\section{An algorithm for bounded looseness}
\label{sec:l}
\looseness=-1 In the previous section, we have seen that \STW{\loosness} for any~$\loosness>1$ is strongly W[1]-hard parameterized by~$m$ and weakly NP-hard for~$m=2$.  %
We complement this result by the following theorem, which yields a \ppt{} algorithm for each constant~$m$ and~$\loosness$.

\begin{theorem}\label{anyalphafpt}
  \STW{\loosness{}} is solvable in $\ell^{O(\loosness m)} \cdot n+O(n\log n)$~time, where $\ell:=\max_{j\in J} |\deadl_j -\relt_j|$.
\end{theorem}

\noindent The crucial observation for the proof of \cref{anyalphafpt} is the following lemma.  It gives a logarithmic upper bound on the height~$\overl{}$ of yes-instances (as defined in \cref{def-st}).  To prove \cref{anyalphafpt},  we will thereafter present an algorithm that has a running time that is single-exponential in~$\overl{}$.

\begin{lemma}\label[lemma]{Stbound}
  Let $I$~be a yes-instance of \STW{\loosness{}} with $m$~machines and $\ell:=\max_{j\in J}|\deadl_j -\relt_j|$.  Then, $I$~has height at most%
\[
2m \cdot \left(\frac{\log \ell}{\log \loosness{} - \log (\loosness{}-1)}+1\right).
\]
\end{lemma}

\begin{proof}
  Recall from \cref{def-st} that the height of an \ICS{} instance is $\max_{t\in\mathbb N}|S_t|$.

  We will show that, in any feasible schedule for~$I$ and at any time~$t$, there are at most $N$~jobs in~$S_t$ that are active on the first machine at some time~$t'\geq t$, where
  \begin{align}
    N\leq 
\frac{\log \ell}{\log \loosness{} - \log (\loosness{}-1)}+1.\label{numberN}
  \end{align}
    By symmetry, there are at most $N$~jobs in~$S_t$ that are active on the first machine at some time~$t'\leq t$.  Since there are $m$~machines, the total number of jobs in~$S_t$ at any time~\(t\), and therefore the height, is at most~$2mN$.  

It remains to show \eqref{numberN}.  To this end,  fix an arbitrary time~\(t\) and an arbitrary feasible schedule for~$I$.  Then, for any~$d\geq 0$, let $J(t+d)\subseteq S_t$~be the set of jobs that are active on the first machine at some time~\(t'\geq t\) but finished by time~$t+d$.  We show by induction on~$d$ that
\begin{align}
|J(t+d)|&\leq\begin{cases} 0 & \text{if }d=0,\\
-\frac{\log d}{ \log (1-{1}/{\loosness{}})}+1&\text{if }d \ge 1.\end{cases}%
                                                                                                   \label{ast}
\end{align}
 If $d=0$, then $|J(t+0)|= 0$ and \eqref{ast} holds.
Now, consider the case $d\geq 1$.  If no job in~$J(t+d)$ is active at time~$t+d-1$, then $J(t+d)=J(t+d-1)$ and \eqref{ast} holds by the induction hypothesis. Now, assume that there is a job~$j\in J(t+d)$ that is active at time~$t+d-1$.  Then, $\deadl_j\geq t+d$ and, since $j\in S_t$, $\relt_j\leq t$.  Hence,
\begin{align*}
\proct_j\geq \frac{|\deadl_j-\relt_j|}\lambda\geq \frac{|t+d-t|}\lambda = \frac d\lambda.
\end{align*}
It follows that 
\begin{align}
  |J(t+d)|&\leq 1+|J(t+d-\lceil d/\lambda\rceil)|.\label{inductstep}
\intertext{Thus, if $d-\lceil d/\loosness\rceil=0$, then $|J(t+d)|\leq1+ |J(t)|\leq 1$ and \eqref{ast} holds.  If $d-\lceil d/\loosness\rceil>0$, then, by the induction hypothesis, the right-hand side of \eqref{inductstep} is}
 \notag{} &\leq 1-\frac{\log (d-\lceil d/\lambda\rceil)}{ \log (1-{1}/{\loosness{}})}+1\\
 \notag{} &\leq 1-\frac{\log( d(1-{1}/{\loosness{}})) }{ \log (1-{1}/{\loosness{}})}+1\\
 \notag{} & =1-\frac{\log d+ \log (1-{1}/{\loosness{}})}{\log (1-{1}/{\loosness{}})}+1\\
 \notag{} &=-\frac{\log d }{ \log (1-{1}/{\loosness{}})}+1,
\end{align}
and \eqref{ast} holds. 
Finally, since $\ell=\max_{1\leq j\leq n} |\deadl_j -\relt_j|$, no job in~$S_t$ is active at time~$t+\ell$.  Hence, we can now prove~\eqref{numberN} using \eqref{ast} by means of
\begin{align*}
   N\leq|J(t+\ell)|&\leq-\frac{\log \ell}{\log (1-{1}/{\loosness{}})}+1\\
  &=-\frac{\log \ell}{\log \bigl(\frac{\loosness{}-1}{\loosness{}}\bigr)}+1\\
  &=-\frac{\log \ell}{\log (\loosness{}-1) - \log \loosness{}}+1\\
  &=\frac{\log \ell}{\log \loosness{} - \log (\loosness{}-1)}+1.\ifspringer\hspace{2cm}\qed\else\qedhere\fi
\end{align*}
\end{proof}

\noindent The following proposition gives some intuition on how the bound behaves for various~$\loosness{}$.
\begin{proposition}\label[proposition]{mlogl}
 For any $\loosness{} \ge 1$ and any $b \in (1,e]$, it holds that
 \begin{align*}
& \frac{1}{\log_b \loosness{} - \log_b (\loosness{}-1)} \le \loosness{}.
 \end{align*}
\end{proposition}
\begin{proof}
  It is well-known that $(1-1/\loosness{})^{\loosness{}} < 1/e$ for any $\loosness{} \ge 1$. Hence, $\loosness{} \log_b (1 -1/\loosness{}) = \log_b (1 -{1}/{\loosness{}})^{\loosness{}} < \log_b {1}/{e} \le -1$, that is, $-\loosness{} \log_b (1 -{1}/{\loosness{}}) \ge 1$. Thus, 
  \begin{align*}
    \frac{1}{-\loosness{} \log_b (1 -{1}/{\loosness{}})} &\le 1 \text{\quad and\quad}\frac{1}{- \log_b (1 -{1}/{\loosness{}})} \le \loosness{}.
  \end{align*}
Finally,
\begin{align*}
\frac{1}{- \log_b (1 -{1}/{\loosness{}})} &= \frac{1}{ -\log_b (\frac{\loosness{}-1}{\loosness{}})}\\
& = \frac{1}{-\log_b (\loosness{}-1) + \log_b \loosness{}}.\ifspringer\hspace{1.6cm}\qed\else\qedhere\fi
\end{align*}
\end{proof}

\noindent\looseness=-1 Towards our proof of \cref{anyalphafpt}, \cref{Stbound} provides a logarithmic upper bound on the height~$\overl{}$ of yes-instances of \ICS{}.  Our second step towards the proof of \cref{anyalphafpt} is the following algorithm, which runs in time that is single-exponential in~$\overl{}$.  We first present the algorithm and, thereafter, prove its correctness and running time.

\begin{algorithm}\label[algorithm]{dpalgo}
  We solve \ICS{} using dynamic programming.  First, for an \ICS{} instance, let $\ell:=\max_{j\in J}|\deadl_j-\relt_j|$, let $S_t$ be as defined in \cref{def-st}, and let $S_t^<\subseteq J$~be the set of jobs~$j$ with $\deadl_j\leq t$, that is, that have to be finished by time~$t$.  

  We compute a table~$T$ that we will show to have the following semantics.  For a time~$t\in\mathbb N$, a subset~$S\subseteq S_t$ of jobs and a vector~$\vec b=(b_1,\dots,b_m)\in \{-\ell,\dots,\ell\}^m$,
  \begin{align*}
        T[t,S,\vec b]&=
                    \begin{cases}
                      1&\text{\parbox[t]{5.5cm}{if all jobs in~$S\cup S_t^<$ can be scheduled so that machine~$i$ is idle from time~$t+b_i$ for each $i \in \{1, \ldots, m\},$}}\\
                      0&\text{otherwise}.
                    \end{cases}
  \end{align*}
  To compute~$T$, first, set $T[0,\emptyset,\vec b]:=1$ for every vector~$\vec b \in \{-\ell,\dots,\ell\}^m$.
  Now we compute the other entries of~$T$ by increasing~$t$, for each~$t$ by increasing~$\vec b$, and for each~$\vec b$ by~$S$ with increasing cardinality.  Herein, we distinguish two cases.
  \begin{enumerate}[(a)]
  \item\label{ca} If $t \ge 1$ and $S \subseteq S_{t-1}$, then set $T[t,S,\vec b]:=T[t-1,S',\vec b']$, where 
    \begin{align*}
      S'&:=S\cup (S_{t-1} \cap S_t^<)\text{ and }
   \vec b':=(b'_1,\dots,b'_m)\text{ with}\\
b'_i&:=\min\{b_i+1,\ell\}\text{ for each~$i \in \{1, \ldots, m\}$}. 
    \end{align*}
  \item\label{cb} Otherwise, set $T[t,S,\vec b]:=1$ if and only if at least one of the following two cases applies:
    \begin{enumerate}[i)]
    \item\label{cbi} there is a machine~$i \in \{1, \dots, m\}$ such that $b_i> -\ell$ and $T[t,S,\vec b'] =1$, where $\vec b':=(b'_1,\dots,b'_m)$ with
      \begin{align*}
        b'_{i'}&:=
                   \begin{cases}
                     b_i-1& \text{if $i'=i$, }\\
                     b_{i'}& \text{if $i'\ne i$, }
                   \end{cases}
      \end{align*}
      or
    \item\label{cbii} there is a job~$j \in S$ and a machine~$i \in \{1, \dots, m\}$ such that $b_i>0$, $t+b_i \le \deadl_j$, $t+b_i-\proct_j\ge \relt_j$, and $T[t,S\setminus \{j\},\vec b']=1$, where $\vec b':=(b'_1,\dots,b'_m)$ with
      \begin{align*}
        b'_{i'}:=
           \begin{cases}
             b_i-\proct_j& \text{if $i'=i$, }\\
             b_{i'}& \text{if $i'\ne i$}.
           \end{cases}
      \end{align*}
      Note that, since $j\in S_t$, one has $\relt_j \ge t-\ell$ by definition of~$\ell$. Hence, $b'_i \ge -\ell$ is within the allowed range~$\{-\ell,\dots,\ell\}$.
    \end{enumerate}
  \end{enumerate}
  Finally, we answer yes if and only if $T[\reltmax, S_{\reltmax}, \mathbf 1^m\cdot \ell]=1$, where $\reltmax := \max_{j \in J}\relt_j$.\qed
\end{algorithm}

\begin{lemma}\label[lemma]{algocorr}
  \cref{dpalgo} correctly decides \ICS{}.
\end{lemma}
\begin{proof} We prove the following two claims: For any time~$0\leq t\leq\relt_{max}$, any set $S\subseteq S_t$, and any vector~$\vec b=(b_1,\dots,b_m)\in\{-\ell,\dots,\ell\}^m$, 
  \begin{align}
    \text{\parbox[t]{7cm}{if $T[t,S,\vec b]=1$, then all jobs in~$S\cup S_t^<$ can be scheduled so that machine~$i$ is idle from time~$t+b_i$ for each $i \in \{1, \ldots, m\}$,}%
       \label{claa}}\\
    \intertext{and}
    \text{\parbox[t]{7cm}{if all jobs in~$S\cup S_t^<$ can be scheduled so that machine~$i$ is idle from time~$t+b_i$ for each $i \in \{1, \ldots, m\}$, then $T[t,S,\vec b]=1$.}}%
      \label{clab}
  \end{align}
  From \eqref{claa} and \eqref{clab}, the correctness of the algorithm easily follows: observe that, in any feasible schedule, all machines are idle from time~$\relt_{max}+\ell$ and all jobs~$J \subseteq S_{\relt_{max}}\cup S_{\relt_{max}}^<$ are scheduled. Hence, there is a feasible schedule if and only if $T[\relt_{max}, S_{\relt_{max}}, \vec 1^m\cdot \ell]=1$.  It remains to prove \eqref{claa} and \eqref{clab}.

  First, we prove \eqref{claa} by induction.  For $T[0,\emptyset,\vec b]=1$, \eqref{claa}~holds since there are no jobs to schedule.  We now prove~\eqref{claa} for~$T[t,S,\vec b]$ under the assumption that it is true for all $T[t',S',\vec b']$ with $t' < t$ or $t'=t$ and $\vec b' \lneqq \vec b$.
 
 If $T[t,S,\vec b]$ is set to~$1$ in \cref{dpalgo}\eqref{ca}, then, for $S'$ and~$\vec b'$ as defined in \cref{dpalgo}\eqref{ca}, $T[t-1,S',\vec b'] =1$.  By the induction hypothesis, all jobs in~$S'\cup S_{t-1}^<$ can be scheduled so that machine~$i$ is idle from time~$t-1+b'_i\le t+b_i$. Moreover, $S\cup S_t^< =S' \cup S_{t-1}^<$ since $S'=S\cup (S_{t-1} \cap S_t^<)$.  Hence, \eqref{claa} follows.
 
 If $T[t,S,\vec b]$ is set to~$1$ in \cref{dpalgo}\eqref{cbi}, then one has $T[t,S,\vec b'] =1$ for $\vec b'$ as defined in \cref{dpalgo}\eqref{cbi}.  By the induction hypothesis, all jobs in~$S\cup S_t^<$ can be scheduled so that machine~$i'$ is idle from time~$t+b'_{i'}\le t+b_{i'}$, and \eqref{claa} follows.
 
 If $T[t,S,\vec b]$ is set to~$1$ in \cref{dpalgo}\eqref{cbii}, then $T[t,S\setminus \{j\},\vec b']=1$ for $j$~and~$\vec b'$ as defined in \cref{dpalgo}\eqref{cbii}.  By the induction hypothesis, all jobs in~$(S\setminus \{j\})\cup S_t^<$ can be scheduled so that machine~$i'$ is idle from time~$t+b'_{i'}$. It remains to schedule job~$j$ on machine~$i$ in the interval~$[t+b'_i, t+b_i)$, which is of length exactly~$\proct_j$ by the definition of~$\vec b'$.  Then, machine~$i$ is idle from time~$t+b_i$ and any machine~$i'\ne i$ is idle from time~$t+b'_{i'}=t+b_{i'}$, and \eqref{claa} follows.
 
 It remains to prove \eqref{clab}.  We use induction.  Claim~\eqref{clab} clearly holds for $t=0$, $S=\emptyset$, and any~$\vec b\in\{-\ell,\dots,\ell\}^m$ by the way \cref{dpalgo} initializes~$T$.  We now show \eqref{clab} provided that it is true for  $t' < t$ or $t'=t$ and $\vec b' \lneqq \vec b$.
 
 If $S\subseteq S_{t-1}$, then $S\cup S_t^< =S' \cup S_{t-1}^<$ for~$S'$ as defined in \cref{dpalgo}\eqref{ca}.  Moreover, since no job in~$S' \cup S_{t-1}^<$ can be active from time~$t-1+\ell$ by definition of~$\ell$, each machine~$i$ is idle from time~$t-1+\min\{b_i+1,\ell\}=t-1+b_i'$, for $\vec b'=(b'_1,\dots,b'_m)$ as defined in \cref{dpalgo}\eqref{ca}.  Hence, $T[t-1,S', \vec b']=1$ by the induction hypothesis, \cref{dpalgo}\eqref{ca} applies, sets $T[t,S,\vec b]:=T[t-1,S',\vec b']=1$, and \eqref{clab} holds.
 
 If some machine~$i$ is idle from time~$t+b_i-1$, then, by the induction hypothesis, $T[t,S,\vec b']=1$  in \cref{dpalgo}\eqref{cbi}, the algorithm sets $T[t,S,\vec b]:=1$, and \eqref{clab} holds.

 In the remaining case, every machine~$i$ is busy at time~$t+b_i-1$ and $K:=S \setminus S_{t-1}\ne\emptyset$.  Thus, there is a machine~$i$ executing a job from~$K$. For each job~$j'\in K$, we have $\relt_{j'} \ge t$.  Since machine~$i$ is idle from time~$t+b_i$ and executes~$j'$, one has $b_i > 0$. Let $j$~be the last job scheduled on machine~$i$. Then, since machine~$i$ is busy at time~$t+b_i-1$, we have $\deadl_j \ge t+b_i > t$ and $j \notin S_t^<$.  Hence, $j \in S_t$.  Since machine~$i$ is idle from time~$t+b_i$, we also have $t+b_i-\proct_j\ge \relt_j$.  Now, if we remove~$j$ from the schedule, then machine~$i$ is idle from time~$t+b_i-\proct_j$ and each machine~$i'\ne i$ is idle from time~$t+b'_{i'}=t+b_{i'}$.  Thus, by the induction hypothesis, $T[t,S\setminus \{j\},\vec b']=1$ in \cref{dpalgo}\eqref{cbii}, the algorithm sets $T[t,S,\vec b]:=1$, and \eqref{clab} holds.\ifspringer\qed\fi
\end{proof}

\begin{lemma}\label[lemma]{lem:algspeed}
  \cref{dpalgo} can be implemented to run in $O( 2^\overl{} \cdot (2\ell+1)^m\cdot (\overl{}^2m+\overl{}m^2)\cdot n\ell+n\log n)$~time, where
    $\ell:=\max_{j\in J}|\deadl_j-\relt_j|$ and $h$~is the height of the input instance. %
\end{lemma}
\begin{proof}
 Concerning the running time of \cref{dpalgo}, we first bound $\relt_{max}$. If $\relt_{max}>n\ell$, then there is a time~$t \in \{0,\dots,\relt_{max}\}$ such that $S_t = \emptyset$ (cf.\ \cref{def-st}). Then, we can split the instance into one instance with the jobs~$S_t^<$ and into one instance with the jobs~$J\setminus S_t^<$.  We answer ``yes'' if and only if both of them are yes-instances.  Henceforth, we assume that $\relt_{max} \le n\ell$.
 
 In a preprocessing step, we compute the sets~$S_t$ and $S_{t-1}\cap S_t^<$, which can be done in $O(n\log n+\overl{}n+\reltmax)$~time by sorting the input jobs by deadlines and scanning over the input time windows once: if no time window starts or ends at time~$t$, then $S_t$~is simply stored as a pointer to the~$S_{t'}$ for the last time~$t'$ where a time window starts or ends.

 Now, the table~$T$ of \cref{dpalgo} has at most $(\relt_{max}+1) \cdot 2^\overl{} \cdot (2\ell+1)^m\leq (n\ell+1) \cdot 2^\overl{} \cdot (2\ell+1)^m$ entries.  A table entry~$T[t,S,\vec b]$~can be accessed in $O(m+\overl)$~time using a carefully initialized trie data structure \citep{Bev13} since $|S|\leq\overl{}$ and since $\vec b$ is a vector of length~$m$.

 To compute an entry~$T[t,S,\vec b]$, we first check, for each job~$j \in S$, whether $j \in S_{t-1}$.  If this is the case for each~$j$, then \cref{dpalgo}\eqref{ca} applies.  We can prepare~$\vec b'$ in $O(m)$~time and $S'$ in $O(\overl{})$~time using the set~$S_{t-1} \cap S_t^<$ computed in the preprocessing step.  Then, we access the entry~$T[t-1,S',\vec b']$ in $O(\overl{}+m)$~time. Hence, \eqref{ca} takes $O(\overl{}+m)$ time.
 
 If \cref{dpalgo}\eqref{ca} does not apply, then we check whether \cref{dpalgo}\eqref{cbi} applies. To this end, for each $i\in\{1,\dots,\allowbreak m\}$, we prepare~$\vec b'$ in $O(m)$~time and access $T[t,S,\vec b']$ in $O(\overl{}+m)$ time. Hence, it takes $O(m^2+\overl{}m)$~time  to check~\eqref{cbi}.
 
 To check whether \cref{dpalgo}\eqref{cbii} applies, we try each~$j \in S$ and each~$i \in \{1, \ldots, m\}$ and, for each, prepare~$\vec b'$ in $O(m)$~time and check $T[t,S\setminus\{j\},\vec b']$ in $O(\overl{}+m)$~time. Thus \eqref{cbii} can be checked in $O(\overl{}^2m + \overl{}m^2)$~time.\ifspringer\qed\fi
\end{proof}

\noindent With the logarithmic upper bound on the height~$\overl{}$ of yes-instances of \ICS{} given by \cref{Stbound} and using \cref{dpalgo}, which, by \cref{lem:algspeed}, runs in time that is single-exponential in~$\overl{}$ for a fixed number~$m$ of machines, we can now prove \cref{anyalphafpt}.%

\begin{proof}[\ifspringer\else{}Proof \fi{}of \cref{anyalphafpt}]
  We use the following algorithm.   Let
  \begin{align*}
 \overl{}:=2m \cdot \left(\frac{\log \ell}{\log \loosness{} - \log (\loosness{}-1)}+1\right).
  \end{align*}
  If, for any time~$t\in\mathbb N$, we have $|S_t|>\overl{}$, then we are facing a no-instance by \cref{Stbound} and immediately answer ``no''.  This can be checked in $O(n\log n)$~time: one uses the interval graph coloring problem to check whether we can schedule the time windows of all jobs (as intervals) onto $\overl{}$~machines.

  Otherwise, we conclude that our input instance has height at most~$h$.
  We now apply \cref{dpalgo}, which, by \cref{lem:algspeed}, runs in $O( 2^\overl{} \cdot (2\ell+1)^m\cdot (\overl{}^2m+\overl{}m^2)\cdot n\ell+n\log n)$~time.  Since, by \cref{mlogl}, $\overl{}\in O(\loosness m \log \ell)$, this running time is $\ell^{O(\loosness{} m)}\overl{}\cdot n+O(n\log n)$. \ifspringer\qed\fi
\end{proof}

\noindent A natural question is whether \cref{anyalphafpt} can be generalized to~$\loosness{}=\infty$, that is, to \ICS{} without looseness constraint.  This question can be easily answered negatively using a known reduction from \tpart{} to \ICS{} given by \citet{GJ79}:

\begin{proposition}
  If there is an $\ell^{O(m)}\cdot\poly(n)$-time algorithm for \ICS{}, where $\ell:=\max_{j\in J} |\deadl_j-\relt_j|$, then P${}={}$NP.
\end{proposition}

\begin{proof}
  \citet[Theorem~4.5]{GJ79} showed that \ICS{} is NP-hard even on~$m=1$ machine.  In their reduction, $\ell\in\poly(n)$.  A supposed $\ell^{O(m)}\cdot\poly(n)$-time algorithm would solve such instances in polynomial time.\ifspringer\qed\fi
\end{proof}

\section{An algorithm for bounded slack}\label{sec:s}

So far, we considered \ICS{} with bounded looseness~$\loosness{}$.  \citet{CEHBW04} additionally considered \ICS{} for any constant slack~$\slack{}$.

\looseness=-1 Recall that \citet{CEHBW04} showed that \STW{\loosness} is NP-hard for any constant~$\loosness>1$ and that \cref{w1hard} shows that having a small number~$m$ of machines does make the problem significantly easier.

Similarly, \citet{CEHBW04} showed that \ICSS{\slack} is NP-hard already for~$\slack=2$.  Now we contrast this result by showing that \ICSS{\slack} is \fp{} tractable for parameter~$m+\slack$.  More specifically, we show the following:

\begin{theorem}\label[theorem]{fpt}
  \ICSS{\slack} is solvable in time
  \begin{align*}
    O\Bigl((\sigma+1)^{(2\sigma+1)m}\cdot n\cdot \sigma m\cdot\log \sigma m + n\log n\Bigr).
  \end{align*}
\end{theorem}

  \noindent Similarly as in the proof of \cref{anyalphafpt}, we first give an upper bound on the height of yes-instances of \ICS{} as defined in \cref{def-st}.  To this end, we first show that each job~$j\in S_t$ has to occupy some of the (bounded) machine resources around time~$t$.

\begin{lemma}\label[lemma]{proc-in-int}
  At any time~$t$ in any feasible schedule for \ICSS{\slack}, each job $j\in S_t$ is active at some time in the interval~$[t-\slack,t+\slack]$.
\end{lemma}

\begin{proof}
  If the time window of~$j$ is entirely contained in~$[t-\slack,t+\slack]$, then, obviously, $j$~is active at some time during the interval~$[t-\slack,t+\slack]$.

  Now, assume that the time window of~$j$ is not contained in~$[t-\slack,t+\slack]$.  Then, since $j\in S_t$, its time window contains~$t$ by \cref{def-st} and, therefore, one of $t-\slack$ or~$t+\slack$.  Assume, for the sake of contradiction, that there is a schedule such that $j$~is not active during~$[t-\slack,t+\slack]$.  Then $j$~is inactive for at least $\slack+1$~time units in its time window---a~contradiction.\ifspringer\qed\fi
\end{proof}

\noindent Now that we know that each job in~$S_t$ has to occupy machine resources around time~$t$, we can bound the size of~$S_t$ in the amount of resources available around that time.

\begin{lemma}\label[lemma]{yesoverlap}
  Any yes-instance of \ICSS{\slack} has height at most~$(2\slack+1)m$.
\end{lemma}

\begin{proof}
  Fix any feasible schedule for an arbitrary yes-instance of \ICSS{\slack} and any time~\(t\). By \cref{proc-in-int}, each job in~$S_t$ is active at some time in the interval~$[t-\slack,t+\slack]$.  This interval has length~$2\slack+1$.  Thus, on $m$~machines, there is at most $(2\slack+1)m$~available processing time in this time interval.  Consequently, there can be at most $(2\slack+1)m$ jobs with time intervals in~$S_t$.\ifspringer\qed\fi
\end{proof}

\noindent We finally arrive at the algorithm to prove \cref{fpt}.

\begin{proof}[\ifspringer\else{}Proof \fi{}of \cref{fpt}]
  Let $\overl{}:=(2\slack+1)m$.  In the same way as for \cref{anyalphafpt}, in $O(n\log n)$~time we discover that we face a no-instance due to \cref{yesoverlap} or, otherwise, that our input instance has height at most~$h$.  In the latter case, we apply the $O(n\cdot (\slack+1)^\overl{}\cdot \overl{}\log \overl{})$-time algorithm due to \citet{CEHBW04}.\ifspringer\qed\fi
\end{proof}

\section{Conclusion}
Despite the fact that there are comparatively few studies on the parameterized complexity of scheduling problems, the field of scheduling indeed offers many natural parameterizations and fruitful challenges for future research. Notably, \citet{Mar11} saw one reason for the lack of results on ``parameterized scheduling'' in the fact that most scheduling problems remain NP-hard even for a constant number of machines (a very obvious and natural parameter indeed), hence destroying hope for \fp{} tractability results with respect to this parameter.  In scheduling interval-constrained jobs with small looseness and small slack, we also have been confronted with this fact, facing (weak) NP-hardness even for two machines.  

The natural way out of this misery, however, is to consider parameter combinations, for instance combining the parameter number of machines with a second one. In our study, these were combinations with looseness and with slack (see also \cref{tab:results}).  In a more general perspective, this consideration makes scheduling problems a prime candidate for offering a rich set of research challenges in terms of a multivariate complexity analysis~\citep{FJR13,Nie10}.  Herein, for obtaining positive algorithmic results, research has to go beyond canonical problem parameters, since basic scheduling problems remain NP-hard even if canonical parameters are simultaneously bounded by small constants, as demonstrated by \citet{KSS12}.\footnote{\looseness=-1The results of \citet{KSS12} were obtained in context of a multivariate complexity analysis framework described by \citet{Sev04}, which is independent of the framework of parameterized complexity theory considered in our work: it allows for systematic classification of problems as polynomial-time solvable or NP-hard given concrete constraints on a set of instance parameters.  It is plausible that this framework is applicable to classify problems as FPT or W[1]-hard as well.}  %

Natural parameters to be studied in future research on \ICS{} are
the combination of slack and looseness---the open field in our \cref{tab:results}---and the maximum and minimum processing times, which were found to play an important role in the online version of the problem \citep{Sah13}.

Finally, we point out that our fixed-parameter algorithms for \ICS{} are easy to implement and may be practically applicable if the looseness, slack, and number of machines is small (about three or four each).  Moreover, our algorithms are based on upper bounds on the height of an instance in terms of its number of machines, its looseness, and slack.  Obviously, this can also be exploited to give lower bounds on the number of required machines based on the structure of the input instance, namely, on its height, looseness, and slack.  These lower bounds may be of independent interest in exact branch and bound or approximation algorithms for the machine minimization problem.

\ifspringer
\bibliographystyle{spbasic}
\else
\bibliographystyle{abbrvnat}
\fi
\bibliography{tw-scheduling}

\newcommand{\noopsort}[1]{}
\begin{thebibliography}{25}
\providecommand{\natexlab}[1]{#1}
\providecommand{\url}[1]{{#1}}
\providecommand{\urlprefix}{URL }
\expandafter\ifx\csname urlstyle\endcsname\relax
  \providecommand{\doi}[1]{DOI~\discretionary{}{}{}#1}\else
  \providecommand{\doi}{DOI~\discretionary{}{}{}\begingroup
  \urlstyle{rm}\Url}\fi
\providecommand{\eprint}[2][]{\url{#2}}

\bibitem[{{\noopsort{Bevern}van Bevern}(2014)}]{Bev13}
{\noopsort{Bevern}van Bevern} R (2014) Towards optimal and expressive
  kernelization for {$d$-Hitting Set}. Algorithmica 70(1):129--147

\bibitem[{{\noopsort{Bevern}van Bevern}
  et~al(2015{\natexlab{a}}){\noopsort{Bevern}van Bevern}, Chen, Hüffner,
  Kratsch, Talmon, and Woeginger}]{BJH+15}
{\noopsort{Bevern}van Bevern} R, Chen J, Hüffner F, Kratsch S, Talmon N,
  Woeginger GJ (2015{\natexlab{a}}) Approximability and parameterized
  complexity of multicover by $c$-intervals. Information Processing Letters
  115(10):744--749

\bibitem[{{\noopsort{Bevern}van Bevern}
  et~al(2015{\natexlab{b}}){\noopsort{Bevern}van Bevern}, Mnich, Niedermeier,
  and Weller}]{BMNW15}
{\noopsort{Bevern}van Bevern} R, Mnich M, Niedermeier R, Weller M
  (2015{\natexlab{b}}) Interval scheduling and colorful independent sets.
  Journal of Scheduling 18(5):449--469

\bibitem[{Bodlaender and Fellows(1995)}]{BF95}
Bodlaender HL, Fellows MR (1995) W[2]-hardness of precedence constrained
  $k$-processor scheduling. Operations Research Letters 18(2):93--97

\bibitem[{Chen et~al(2016)Chen, Megow, and Schewior}]{CMS16}
Chen L, Megow N, Schewior K (2016) An {$O(\log m)$}-competitive algorithm for
  online machine minimization. In: Proceedings of the 27th Annual ACM-SIAM
  Symposium on Discrete Algorithms (SODA'16), SIAM, pp 155--163

\bibitem[{Chuzhoy et~al(2004)Chuzhoy, Guha, Khanna, and Naor}]{CGKN04}
Chuzhoy J, Guha S, Khanna S, Naor J (2004) Machine minimization for scheduling
  jobs with interval constraints. In: Proceedings of the 45th Annual Symposium
  on Foundations of Computer Science (FOCS'04), pp 81--90

\bibitem[{{Cie\-lie\-bak} et~al(2004){Cie\-lie\-bak}, Erlebach, Hennecke,
  Weber, and Widmayer}]{CEHBW04}
{Cie\-lie\-bak} M, Erlebach T, Hennecke F, Weber B, Widmayer P (2004)
  Scheduling with release times and deadlines on a minimum number of machines.
  In: Exploring New Frontiers of Theoretical Informatics, IFIP International
  Federation for Information Processing, vol 155, Springer, pp 209--222

\bibitem[{Cygan et~al(2015)Cygan, Fomin, Kowalik, Lokshtanov, Marx, Pilipczuk,
  Pilipczuk, and Saurabh}]{CFK+15}
Cygan M, Fomin FV, Kowalik L, Lokshtanov D, Marx D, Pilipczuk M, Pilipczuk M,
  Saurabh S (2015) Parameterized Algorithms. Springer

\bibitem[{Downey and Fellows(2013)}]{DF13}
Downey RG, Fellows MR (2013) Fundamentals of Parameterized Complexity. Springer

\bibitem[{Fellows and McCartin(2003)}]{FM03}
Fellows MR, McCartin C (2003) On the parametric complexity of schedules to
  minimize tardy tasks. Theoretical Computer Science 298(2):317--324

\bibitem[{Fellows et~al(2013)Fellows, Jansen, and Rosamond}]{FJR13}
Fellows MR, Jansen BMP, Rosamond FA (2013) Towards fully multivariate
  algorithmics: Parameter ecology and the deconstruction of computational
  complexity. European Journal of Combinatorics 34(3):541--566

\bibitem[{Flum and Grohe(2006)}]{FG06}
Flum J, Grohe M (2006) Parameterized Complexity Theory. Springer

\bibitem[{Garey and Johnson(1979)}]{GJ79}
Garey MR, Johnson DS (1979) Computers and Intractability: A Guide to the Theory
  of {NP}-Completeness. Freeman

\bibitem[{Halld{\'o}rsson and Karlsson(2006)}]{HK06}
Halld{\'o}rsson MM, Karlsson RK (2006) Strip graphs: Recognition and
  scheduling. In: Proceedings of the 32nd International Workshop on
  Graph-Theoretic Concepts in Computer Science (WG'06), Springer, LNCS, vol
  4271, pp 137--146

\bibitem[{Hermelin et~al(2015)Hermelin, Kubitza, Shabtay, Talmon, and
  Woeginger.}]{HKS+15}
Hermelin D, Kubitza JM, Shabtay D, Talmon N, Woeginger G (2015) Scheduling two
  competing agents when one agent has significantly fewer jobs. In: Proceedings
  of the 10th International Symposium on Parameterized and Exact Computation
  (IPEC'15), Leibniz International Proceedings in Informatics (LIPIcs), vol~43,
  Schloss Dagstuhl--Leibniz-Zentrum für Informatik, pp 55--65

\bibitem[{Jansen et~al(2013)Jansen, Kratsch, Marx, and Schlotter}]{JKMS13}
Jansen K, Kratsch S, Marx D, Schlotter I (2013) Bin packing with fixed number
  of bins revisited. Journal of Computer and System Sciences 79(1):39--49

\bibitem[{Kolen et~al(2007)Kolen, Lenstra, Papadimitriou, and
  Spieksma}]{KLPS07}
Kolen AWJ, Lenstra JK, Papadimitriou CH, Spieksma FCR (2007) Interval
  scheduling: A survey. Naval Research Logistics 54(5):530--543

\bibitem[{Kononov et~al(2012)Kononov, Sevastyanov, and Sviridenko}]{KSS12}
Kononov A, Sevastyanov S, Sviridenko M (2012) A complete 4-parametric
  complexity classification of short shop scheduling problems. Journal of
  Scheduling 15(4):427--446

\bibitem[{Malucelli and Nicoloso(2007)}]{MN07}
Malucelli F, Nicoloso S (2007) Shiftable intervals. Annals of Operations
  Research 150(1):137--157

\bibitem[{Marx(2011)}]{Mar11}
Marx D (2011) Fixed-parameter tractable scheduling problems. In: Packing and
  Scheduling Algorithms for Information and Communication Services (Dagstuhl
  Seminar 11091)

\bibitem[{Mnich and Wiese(2015)}]{MW15}
Mnich M, Wiese A (2015) Scheduling and fixed-parameter tractability.
  Mathematical Programming 154(1-2):533--562

\bibitem[{Niedermeier(2006)}]{Nie06}
Niedermeier R (2006) Invitation to Fixed-Parameter Algorithms. Oxford
  University Press

\bibitem[{Niedermeier(2010)}]{Nie10}
Niedermeier R (2010) Reflections on multivariate algorithmics and problem
  parameterization. In: Proceedings of the 27th~International Symposium on
  Theoretical Aspects of Computer Science (STACS'10), Schloss
  Dagstuhl--Leibniz-Zentrum f{\"u}r Informatik, Leibniz International
  Proceedings in Informatics (LIPIcs), vol~5, pp 17--32

\bibitem[{Saha(2013)}]{Sah13}
Saha B (2013) Renting a cloud. In: Annual Conference on Foundations of Software
  Technology and Theoretical Computer Science ({FSTTCS}) 2013, Schloss
  Dagstuhl--Leibniz-Zentrum f\"{u}r Informatik, Leibniz International
  Proceedings in Informatics (LIPIcs), vol~24, pp 437--448

\bibitem[{Sevastianov(2005)}]{Sev04}
Sevastianov SV (2005) An introduction to multi-parameter complexity analysis of
  discrete problems. European Journal of Operational Research 165(2):387--397

\end{thebibliography}
\end{document}